\documentclass[journal]{IEEEtran}
\IEEEoverridecommandlockouts                              

\usepackage{cite}
\usepackage{url}
\usepackage{amsmath}
\usepackage{amssymb}
\usepackage{amsthm}
\usepackage{mathrsfs}
\usepackage{mathtools}
\usepackage{bm}
\usepackage{algorithm}
\usepackage{algorithmic}
\usepackage{graphicx}
\usepackage{subcaption}
\usepackage{xcolor}
\usepackage{booktabs}
\usepackage{tikz}
\usetikzlibrary{shapes,arrows,positioning,calc}

\newtheorem{theorem}{Theorem}
\newtheorem{lemma}{Lemma}

\newtheorem{problem}{Problem}

\DeclareMathOperator{\diag}{diag}

\DeclareMathOperator{\CUT}{CUT}
\DeclareMathOperator{\MAXCUT}{MAXCUT}

\title{APX-Hardness of Computing Lipschitz Constants for Multi-Parametric Quadratic Programs}

\author{Xingchen Li, Kunpeng Liu, Keyou You%
\thanks{This work was supported by National Natural Science Foundation of China (62325305), and the BNRist project (No. BNR2024TD03003). (Corresponding author: Keyou You)}%
\thanks{Xingchen Li, Kunpeng Liu, and Keyou You are with the Department of Automation, and Beijing National Research Center for Info. Sci. \& Tech. (BNRist), Tsinghua University, Beijing 100084, China. (e-mail: \{lixc21,liukp22\}@mails.tsinghua.edu.cn, youky@tsinghua.edu.cn)}%
}
\date{\today}

\begin{document}
\maketitle
\thispagestyle{empty}
\pagestyle{empty}

\begin{abstract}
Computing the Lipschitz constant of the solution map of a multi-parametric quadratic program is important for the analysis of optimization-based control. This problem is governed by three factors: the parameter dimension, the number of decision variables, and the number of constraints. While empirical evidence has long suggested exponential complexity, a rigorous complexity-theoretic proof has been lacking. In this paper, we fill this gap by proving that this problem is not only NP-hard but also APX-hard. Furthermore, we reveal that: (a) the problem becomes polynomial-time solvable when the number of constraints or decision variables is fixed; and (b) both NP-hardness and APX-hardness persist even in the scalar parameter case. These results confirm that the complexity stems from the number of constraints and variables, rather than the parameter dimension. Numerical experiments further validate these theoretical findings.
\end{abstract}


\section{Introduction}

Model predictive control (MPC) is a widely adopted control strategy in industrial applications, including chemical processes \cite{chen2012distributed}, power systems \cite{venkat2008distributed}, robotics \cite{nubert2020safe}, and autonomous vehicles \cite{cheng2019longitudinal}. At each sampling instant, MPC determines the optimal control input by solving a quadratic program (QP) parameterized by the current system state. This procedure naturally defines a multi-parametric quadratic program (mpQP), where the optimal control input is expressed as a function of the state \cite{Bemporad2002}.

The Lipschitz constant of the mpQP solution map --- quantifying the maximum rate of change of the optimal input with respect to state perturbations --- plays a central role across optimization and control theory. From a parametric optimization perspective, it is a fundamental regularity measure of the solution map, directly connected to stability moduli in sensitivity analysis \cite{Fiacco1983} and to the classical Hoffman error bound for polyhedral systems \cite{pena2019equivalences}. In control, a bounded Lipschitz constant ensures asymptotic stability under estimation errors \cite{scokaert1997discrete} and enables robust invariant set constructions \cite{marruedo2002stability}. More broadly, it determines the disturbance amplification factor in tube-based robust MPC \cite{Mayne2005}, underpins recursive feasibility and stability guarantees of MPC under inexact numerical optimization \cite{Rubagotti2014}, and quantifies the performance degradation of suboptimal real-time MPC \cite{Zeilinger2011}. It also enables reliable constraint removal strategies for reducing mpQP complexity \cite{hou2024harnessing}.

It is well-established that the solution map of an mpQP is continuous and piecewise affine \cite{Bemporad2002}. This structure arises because the set of active constraints remains invariant within polyhedral regions of the parameter space, known as critical regions (CRs). Consequently, the Lipschitz constant is determined by the maximum spectral norm of the affine function's Jacobian across all CRs. A direct approach therefore requires enumerating all CRs, which is central to explicit MPC \cite{Borrelli2017}. However, the number of CRs is empirically observed to grow exponentially with the problem size \cite{Alessio2009}. 

To bypass the full enumeration, existing approaches can be broadly categorized into two classes. \emph{Explicit enumeration} methods \cite{Bemporad2002} construct the complete polyhedral partition; decomposition strategies \cite{Oberdieck2016} and variable separation techniques \cite{Kvasnica2015} can reduce the subproblem size but do not overcome the exponential growth of CRs. \emph{Mixed-integer programming} formulations \cite{fabiani2022reliably} compute exact Lipschitz constants or tight bounds without full enumeration, yet may remain exponential in the worst case and typically rely on the 1-norm or $\infty$-norm, necessitating conservative conversions for the 2-norm required in Lyapunov analysis. 

Despite these advances, a foundational question remains open: the computational complexity of computing the Lipschitz constant for a general mpQP is unclear. Although the problem is widely presumed to be intractable, no rigorous theoretical proof has been established. Furthermore, it remains unclear whether the computational difficulty stems from the parameter dimension, the number of constraints, the number of decision variables, or the interplay between them.

To bridge this gap, we provide the first rigorous complexity-theoretic analysis of computing the Lipschitz constant of mpQPs. Via a gap-preserving reduction from MAX-CUT \cite{Karp1972, haastad2001some}, we demonstrate that this computation is not only NP-hard but also APX-hard, ruling out any polynomial-time approximation scheme unless P$=$NP. Furthermore, we characterize the complexity with respect to problem dimensions, showing that both NP-hardness and APX-hardness persist even for scalar parameters. This result is established through a reduction chain from MAX-2SAT \cite{garey1979computers} to the Lipschitz constant of a scalar-input ReLU network, and subsequently to an mpQP instance. This confirms that the computational intractability is driven by the number of constraints and decision variables, whereas the problem becomes polynomial-time solvable when either of these dimensions is fixed.

The remainder of this paper is organized as follows. Section~II introduces the mpQP and the necessary background. Section~III establishes the NP-hardness and APX-hardness of the problem. Section~IV analyzes the complexity with respect to problem dimensions. Section~V presents numerical experiments. Section~VI concludes the paper. 

In comparison with our conference version \cite{li2026conference}, we further add numerical experiments, refine the scalar-parameter hardness proof with Gray-code encoding, and provide explicit polynomial-time constructions from ReLU networks to mpQPs along with implementation references.

\section{Problem Formulation and Preliminaries}\label{sec:problem}

We first define the mpQP and derive the piecewise affine structure of its solution map via KKT conditions. We then introduce the Lipschitz constant of the solution map, recall the relevant notions from complexity theory, and state the objectives of this paper.

\subsection{Multi-Parametric Quadratic Programs and Lipschitz Continuity}
We consider the strictly convex multi-parametric quadratic program (mpQP) defined as:
\begin{equation} \label{eq:mpqp_def}
    \min_{z \in \mathbb{R}^{n_z}} \ \frac{1}{2} z^\top Q z + x^\top H z + f^\top z \quad \text{s.t.} \quad A z \leq b + S x,
\end{equation}
where $Q \in \mathbb{R}^{n_z \times n_z}$ with $Q \succ 0$, $H \in \mathbb{R}^{n_x \times n_z}$, $f \in \mathbb{R}^{n_z}$, $A \in \mathbb{R}^{n_c \times n_z}$, $b \in \mathbb{R}^{n_c}$, and $S \in \mathbb{R}^{n_c \times n_x}$.
The feasible set $\mathcal{F}(x) = \{z \in \mathbb{R}^{n_z} \mid A z \leq b + S x\}$ is assumed non-empty for all $x$ in the polyhedral parameter domain $\mathcal{X} \subseteq \mathbb{R}^{n_x}$. Strict convexity guarantees that the optimizer is unique for each $x \in \mathcal{X}$, thereby defining the \emph{solution map} $z^*: \mathcal{X} \to \mathbb{R}^{n_z}$.

Since \eqref{eq:mpqp_def} is convex with linear constraints, the KKT conditions are necessary and sufficient for optimality. Stationarity gives $Q z + H^\top x + f + A^\top \lambda = 0$, and complementarity $\lambda_i (A_i z - b_i - S_i x) = 0$ partitions the constraints at each $x$ into an \emph{active set} $\mathcal{A}(x) := \{i : A_i z^*(x) = b_i + S_i x\}$ and its complement, with $\lambda_i = 0$ for all $i \notin \mathcal{A}(x)$. By standard mpQP theory \cite{Borrelli2017}, the parameter space is partitioned into a finite set of polyhedral \emph{critical regions} $\{\mathcal{R}_i\}_{i\in\mathcal I}$, each associated with a fixed active set $\mathcal A_i$, on which the optimal solution is affine:
\begin{equation*}
    z^*(x) = K_i x + c_i \quad \text{on } \mathcal R_i.
\end{equation*}
For each critical region, the affine gain $K_i$ can be generated from a linearly independent subset $\mathcal{A}\subseteq\mathcal{A}_i$ with $\mathrm{rank}(A_{\mathcal{A}})=|\mathcal{A}|\le n_z$ \cite{darup2017maximal}, so every gain $K_i$ corresponds to some full-rank active subset. We study the Lipschitz constant of the solution map $z^*$:
\begin{equation}\label{eq:lipschitz_def}
    L := \max_{i \in \mathcal{I}} \|K_i\|_2.
\end{equation}

A natural approach to computing $L$ is to enumerate the critical regions and evaluate $\|K_i\|_2$ on each. However, this approach is fundamentally limited by the number of regions. Since each gain $K_i$ is generated by a linearly independent active subset of cardinality at most $n_z$ \cite{darup2017maximal}, the number of distinct gains is bounded by $\sum_{j=0}^{\min(n_c,n_z)}\binom{n_c}{j}$, which is exponential in the worst case. Enumeration therefore becomes intractable as the problem grows. This explains the practical difficulty of enumeration-based methods, but leaves open the more fundamental question of whether $L$ can be computed in polynomial time by \emph{any} algorithm. Answering this requires the language of computational complexity, which we now briefly review.

\subsection{Complexity Classes and Computational Hardness}
For readers without a complexity-theory background, we briefly recall the two notions of hardness used in our analysis \cite{Arora2009}: \emph{NP-hardness}, which rules out exact polynomial-time algorithms, and \emph{APX-hardness}, which further rules out arbitrarily accurate polynomial-time approximations.

\paragraph*{NP-Hardness} The class \emph{NP} consists of decision problems for which any proposed solution can be verified in polynomial time. It contains many classical combinatorial problems, such as Boolean satisfiability (SAT), the travelling salesman problem (TSP), and MAX-CUT \cite{Karp1972}, none of which is known to admit a polynomial-time algorithm despite decades of effort. A problem is \emph{NP-hard} if every problem in NP can be reduced to it in polynomial time, so that an efficient algorithm for it would yield one for all of NP. Under the widely believed conjecture $\text{P}\neq\text{NP}$, no NP-hard problem admits a polynomial-time exact algorithm. 

\paragraph*{APX-Hardness} When exact computation is intractable, one may instead seek an approximate solution. For a maximization problem with optimum $L$, an \emph{$\rho$-approximation algorithm} ($\rho\ge 1$) is a polynomial-time algorithm that returns a value $\widehat L$ with $\widehat L \le L \le \rho\,\widehat L$. A problem is called \emph{APX-hard} if there is a fixed threshold $\rho^*>1$ such that $\rho$-approximation is NP-hard for every $\rho<\rho^*$. Hence APX-hardness implies NP-hardness. It also rules out any \emph{polynomial-time approximation scheme} (PTAS), i.e., a family of polynomial-time algorithms attaining ratio arbitrarily close to~$1$.

\subsection{The Objective of This Paper}

The exponential bound on critical regions makes enumeration impractical, but does not by itself preclude a more efficient algorithm. Equipped with the language of complexity, we can now ask the sharper questions that drive this paper:
\begin{enumerate}
    \item Is computing $L$ inherently hard, i.e., NP-hard or even APX-hard?
    \item How do the dimensions $n_z$, $n_c$, and $n_x$ each contribute to this hardness?
\end{enumerate}

This paper resolves both questions: we prove that computing $L$ is APX-hard (and hence NP-hard), and we further characterize how this hardness depends on the problem dimensions. The APX-hardness is established by showing that the following approximation problem is NP-hard.

\begin{problem}[$\rho$-approximation of mpQP-Lipschitz]\label{prob:lip_approx}
Given mpQP data $(Q,H,f,A,b,S)$ with $Q\succ 0$ over a polyhedral domain $\mathcal{X}$, and an approximation ratio $\rho\ge 1$, compute a $\widehat{L}$ such that $\widehat{L}\le L\le \rho\widehat{L}$. 
\end{problem}

The hardness of Problem~\ref{prob:lip_approx} is established in two parts. Section~\ref{sec:np_hardness} proves the general APX-hardness via a reduction from MAX-CUT. Section~\ref{sec:dimensional_analysis_new} then analyzes the role of each problem dimension: it identifies tractable cases when $n_c$ or $n_z$ is fixed (Section~\ref{subsec:tractable_new}), and shows via a reduction from MAX-2SAT that the hardness persists even for $n_x=1$ (Section~\ref{subsec:scalar_new}). The two reduction paths are summarized in Fig.~\ref{fig:reduction_roadmap}.

\begin{figure}[!t]
\centering
\begin{tikzpicture}[node distance=0.8cm,
    block/.style={rectangle, draw, thick, text width=2.8cm, text centered, rounded corners, minimum height=0.75cm, font=\scriptsize},
    arrow/.style={->, thick}]
    \node[block, fill=red!10] (maxcut) {MAX-CUT\\(Known APX-hard)\cite{haastad2001some}};

    \node[block, fill=red!10, right=1.5cm of maxcut] (sat) {MAX-2SAT\\(Known APX-hard)\cite{haastad2001some}};
    \node[block, fill=orange!10, below=of sat] (deep) {1D-NN-Lipschitz\\};
    \draw[arrow] (sat) -- node[right, font=\tiny, align=center] {Lem.~\ref{lem:1d_nn_hard}\\(Sec.~\ref{subsec:scalar_new})} (deep);

    \node[block, fill=green!10, text width=2.8cm, below=1.2cm of $(maxcut |- deep)!0.5!(deep)$] (mpqp) {mpQP-Lipschitz\\ (APX-hard)};
    \draw[arrow] (maxcut) -- node[left, font=\tiny, align=center] {Thm.~\ref{thm:complexity_main}\\(Sec.~\ref{sec:np_hardness})} (mpqp);
    \draw[arrow] (deep) -- node[right, xshift=0.2cm, font=\tiny, align=center] {Lem.~\ref{lem:mpqp_relu_equiv}\\+Thm.~\ref{thm:scalar_hard}\\(Sec.~\ref{subsec:scalar_new})} (mpqp);
\end{tikzpicture}
\caption{Reduction roadmap. Left path: a direct reduction from MAX-CUT (known APX-hard \cite{haastad2001some}) establishes that computing the mpQP Lipschitz constant is NP-hard and APX-hard (Thm.~\ref{thm:complexity_main}). Right path: MAX-2SAT (known APX-hard \cite{haastad2001some}) reduces to the Lipschitz constant of a scalar-input ReLU network (Lem.~\ref{lem:1d_nn_hard}), which is then embedded into an mpQP (Lem.~\ref{lem:mpqp_relu_equiv}) and combined with an output-layer scaling argument (Thm.~\ref{thm:scalar_hard}), showing that hardness persists even for $n_x=1$.}
\vspace{-0.4cm}
\label{fig:reduction_roadmap}
\end{figure}
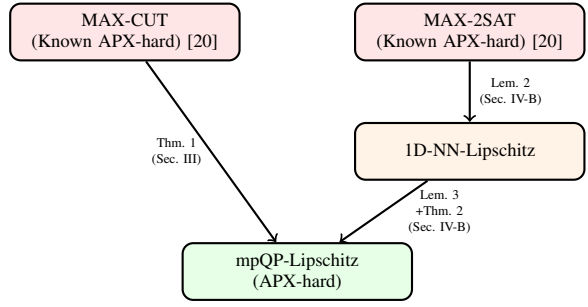

\section{Computing Lipschitz Constants of mpQPs is NP-hard and APX-hard}\label{sec:np_hardness}

We establish the computational intractability of computing $L$ via a gap-preserving reduction from MAX-CUT, which is APX-hard \cite{Papadimitriou1991}.

\begin{theorem}\label{thm:complexity_main}
Computing the Lipschitz constant $L$ of an mpQP is NP-hard and APX-hard.
\end{theorem}

\begin{proof}
We establish the hardness results via a gap-preserving reduction from the MAX-CUT problem. Consider an unweighted connected graph $G=(V, E)$ with $n$ vertices and $m$ edges. For a partition of vertices into two sets $S$ and $\bar{S} = V \setminus S$, define the indicator vector $s \in \{0, 1\}^n$ by $s_i = 1$ if $i \in S$ and $s_i = 0$ otherwise. The cut value $\CUT(G, s)$ counts the number of edges connecting vertices in different sets. Let $C \in \mathbb{R}^{m \times n}$ be the \textit{oriented incidence matrix} of $G$, where each row corresponds to an edge $(i,j) \in E$ with arbitrary orientation, having entry $+1$ at column $i$ and $-1$ at column $j$. The choice of orientation does not affect the analysis since only the magnitude $\|Cs\|_2$ enters in the sequel. The cut value can be expressed algebraically as $\CUT(G, s) = \|C s\|_2^2 = \sum_{(i,j) \in E} (s_i - s_j)^2$, where each term $(s_i - s_j)^2 \in \{0, 1\}$ equals $1$ if and only if $s_i \neq s_j$. The MAX-CUT problem asks for:
\[
\MAXCUT(G) = \max_{s \in \{0, 1\}^n} \|C s\|_2^2.
\]

We construct an mpQP instance with parameter $x \in \mathbb{R}^m$ and decision variable $z \in \mathbb{R}^n$, corresponding to $n_x=m$, $n_z=n$, and $n_c=n$. Let $\delta > 0$ be a rational perturbation parameter. Define matrices $M_1 = \mathbf{1}\mathbf{1}^\top + \delta I_n$ and $M_2 = C^\top$. The constructed mpQP is:
\begin{equation}\label{eq:mpqp_construction}
    \min_{z} \ \frac{1}{2} z^\top (M_1 M_1^\top)^{-1} z \quad \text{s.t.} \quad M_1^{-1} z \ge M_2 x.
\end{equation}
By substitution $u = M_1^{-1} z$, the optimal solution is $z^*(x) = M_1 \max(0, M_2 x)$. The Jacobian in a region with activation pattern $\Sigma = \diag(s)$ for $s \in \{0, 1\}^n$ is $J_s = M_1 \Sigma M_2$, and the squared Lipschitz constant is $L_\delta^2 = \max_{s \in \mathcal{S}} \|J_s\|_2^2$, where $\mathcal{S}$ denotes the set of feasible activation patterns.

Since $G$ is connected, $\ker(C)=\mathrm{span}\{\mathbf 1\}$ and hence $\mathrm{Im}(C^\top)=\mathbf{1}^\perp$. For any non-constant $s\in\{0,1\}^n$, the choice $y=s-(\mathbf 1^\top s/n)\mathbf 1$ lies in $\mathrm{Im}(C^\top)$ and satisfies $\mathrm{sign}(y)=2s-\mathbf 1$, so there exists $x$ with $M_2 x = y$ realizing the pattern $s$. In particular, the optimal MAX-CUT partition is always feasible.

Substituting $M_1$, we decompose the Jacobian into a dominant term and a perturbation:
\[
J_s = (\mathbf{1}\mathbf{1}^\top + \delta I) \Sigma C^\top = \underbrace{\mathbf{1} s^\top C^\top}_{A_s} + \underbrace{\delta \Sigma C^\top}_{E_s}.
\]
The dominant term $A_s = \mathbf{1}(Cs)^\top$ is a rank-1 matrix with spectral norm $\|A_s\|_2 = \|\mathbf{1}\|_2 \|C s\|_2 = \sqrt{n \cdot \CUT(G, s)}$. The perturbation term satisfies $\|E_s\|_2 \le \delta \|C^\top\|_2 \le \delta \sqrt{2m}$, since each row of $C$ has exactly two nonzero entries $\pm 1$.

By the triangle inequality, $\|J_s\|_2$ lies between $\|A_s\|_2 - \|E_s\|_2$ and $\|A_s\|_2 + \|E_s\|_2$. Squaring these bounds and using $\|A_s\|_2 = \sqrt{n \cdot \CUT(G, s)}$ and $\|E_s\|_2 \le \delta \sqrt{2m}$, we obtain
\[
\left| \|J_s\|_2^2 - n \cdot \CUT(G, s) \right| \le 2\delta \sqrt{2mn \cdot \CUT(G, s)} + 2\delta^2 m.
\]
Applying this bound at the maximizer $s^*\in\arg\max_{s\in\mathcal S}\|J_s\|_2^2$ (upper) and at the MAX-CUT optimizer $s_{\mathrm{MC}}$ (lower, feasible by the preceding paragraph), and using $\CUT\le m$ to simplify, we obtain
\begin{equation}\label{eq:abs_gap}
    \bigl| L_\delta^2 - n\cdot\MAXCUT(G) \bigr| \;\le\; 2\delta m\sqrt{2n} + 2\delta^2 m.
\end{equation}
Fix any $\varepsilon\in(0,1/17)$ and choose the rational $\delta=\varepsilon/(8mn)$ (polynomial bit-length). The right-hand side of \eqref{eq:abs_gap} is then at most $\varepsilon\sqrt{2}/(4\sqrt n)+\varepsilon^2/(32mn^2)<\varepsilon$, hence
\begin{equation}\label{eq:gap_clean}
    \bigl| L_\delta^2/n - \MAXCUT(G) \bigr| \;<\; \varepsilon/n.
\end{equation}

Suppose, for contradiction, that for some fixed $\rho<\sqrt{17/16}$ a polynomial-time algorithm returns a rational $\widehat L$ satisfying $\widehat L\le L_\delta\le\rho\widehat L$. Squaring (both sides nonnegative) gives $\widehat L^2\le L_\delta^2\le\rho^2\widehat L^2$, and combining with \eqref{eq:gap_clean} sandwiches $\widehat L^2/n$ as
\begin{equation}\label{eq:M_bounds}
    \frac{\MAXCUT(G)}{\rho^2} - \frac{\varepsilon}{n\rho^2} \;<\; \frac{\widehat L^2}{n} \;<\; \MAXCUT(G) + \frac{\varepsilon}{n}.
\end{equation}
Set $\Delta:=1/\rho^2-16/17>0$. By H\aa{}stad's theorem \cite{haastad2001some}, for any $\eta>0$ it is NP-hard to distinguish $\MAXCUT(G)\ge\alpha$ from $\MAXCUT(G)\le (16/17+\eta)\alpha$, with $\alpha\ge m/2$. Taking $\eta=\Delta/2$, the YES/NO ranges of $\widehat L^2/n$ from \eqref{eq:M_bounds} are disjoint whenever $(\Delta/2)\alpha>\varepsilon(1+\rho^{-2})/n$, which (using $\alpha\ge m/2$, $1+\rho^{-2}\le 2$) holds for all sufficiently large $mn$ under our choice $\varepsilon<1/17$. Hence the algorithm decides H\aa{}stad's gap in polynomial time, contradicting P$\ne$NP. This proves that Problem~\ref{prob:lip_approx} is NP-hard for every fixed $\rho\in(1,\sqrt{17/16})$. Moreover, since the construction \eqref{eq:mpqp_construction}--\eqref{eq:gap_clean} is a gap-preserving reduction from MAX-CUT, which is APX-hard \cite{Papadimitriou1991}, computing $L$ is APX-hard, and a fortiori NP-hard.
\end{proof}

\section{The Role of Problem Dimensions in Complexity}\label{sec:dimensional_analysis_new}

Having established general NP-hardness and APX-hardness in Section~\ref{sec:np_hardness}, we now investigate which problem dimensions drive this computational intractability. Understanding this is crucial for practitioners: if hardness stems from high-dimensional parameters, low-order systems might be tractable; alternatively, if hardness arises from the number of constraints or decision variables, short-horizon problems might be easier. 

Our analysis reveals an interesting contrast. On one hand, fixing the number of constraints $n_c$ or decision variables $n_z$ renders the problem polynomial-time solvable (Section~\ref{subsec:tractable_new}). On the other hand, hardness persists even when $n_x = 1$ (Section~\ref{subsec:scalar_new}). Together, these results show that the fundamental source of complexity is the \emph{combinatorial explosion of critical regions} as both $n_c$ and $n_z$ grow, rather than the parameter space dimension.

\subsection{Tractable Cases: Fixed $n_c$ or $n_z$}\label{subsec:tractable_new}

We begin by characterizing the special cases where the problem admits polynomial-time algorithms.

\begin{lemma}\label{lem:fixed_nc_nz_new}
If $n_c$ or $n_z$ is fixed, then $L$ is computable in polynomial time.
\end{lemma}

\begin{proof}
By \cite{darup2017maximal}, every gain $K_i$ is generated by a linearly independent active subset $\mathcal{A}\subseteq\{1,\ldots,n_c\}$ with $\mathrm{rank}(A_\mathcal{A})=|\mathcal{A}|\le n_z$. The number of such candidate subsets is bounded by $\sum_{j=0}^{n_z}\binom{n_c}{j}$, which is $O(1)$ when $n_c$ is fixed and $O(n_c^{n_z})$ when $n_z$ is fixed, hence polynomial in either case. For each candidate $\mathcal{A}$, $K_\mathcal{A}$ is evaluated by rational arithmetic in $\mathrm{poly}(n_z,n_x)$ time. Discarding empty regions and taking $L=\max_{\mathcal A}\|K_\mathcal A\|_2$ over the polynomially many survivors yields $L$ in polynomial time.
\end{proof}

\subsection{Hardness Persists for Scalar Parameters ($n_x = 1$)}
\label{subsec:scalar_new}

The tractability results in Section~\ref{subsec:tractable_new} raise a critical question: does fixing the parameter dimension $n_x$ alone guarantee tractability? While the geometric simplicity of a scalar domain ($n_x = 1$) might intuitively suggest low complexity, we demonstrate that this is not the case: both NP-hardness and APX-hardness persist even for $n_x=1$.

The obstacle to a direct reduction is that with only a scalar parameter, the natural combinatorial signal carried by the $2^{n_z}$ binary activation patterns in the MAX-CUT construction of Theorem~\ref{thm:complexity_main} is no longer accessible. We therefore route the hardness through an intermediate object that, despite a scalar input, can still partition the domain into exponentially many affine pieces: a deep ReLU network. Such a network can ``fold'' $[0,1]$ exponentially many times via a sawtooth composition (Fig.~\ref{fig:sawtooth}), producing $2^n$ critical regions from depth-$n$ tent maps. Crucially, every ReLU network is itself realizable as an mpQP solution map (Lemma~\ref{lem:mpqp_relu_equiv}), so any hardness for the network transfers to the mpQP with $n_x=1$.

Our reduction therefore proceeds in two steps (Fig.~\ref{fig:reduction_roadmap}): first from MAX-2SAT to the Lipschitz computation of a scalar-input ReLU network (the \emph{1D-NN-Lipschitz problem}), then from 1D-NN-Lipschitz to mpQP-Lipschitz with $n_x=1$. We choose MAX-2SAT over MAX-CUT as the source problem because its $O(n+m)$ bits of Boolean structure fit naturally into the sawtooth folding: the $2^n$ atomic sub-intervals of an $n$-fold tent composition are in bijection with the $2^n$ Boolean assignments. Formally, the 1D-NN-Lipschitz problem asks for the Lipschitz constant of a scalar-input, scalar-output ReLU network $F: [0,1] \to \mathbb{R}$, defined recursively by $h_0 = x$ and $h_\ell = \mathrm{ReLU}(W_\ell h_{\ell-1} + b_\ell)$ for $\ell=1,\dots,D$, with final linear output $F(x) = W_{D+1} h_D + b_{D+1}$ and rational parameters.

\begin{lemma}\label{lem:1d_nn_hard}
The 1D-NN-Lipschitz problem is APX-hard: unless P$=$NP, no polynomial-time algorithm achieves an approximation ratio below $22/21$.
\end{lemma}

\begin{proof}
We reduce MAX-2SAT to computing the Lipschitz constant of a 1D ReLU network. Let $\phi$ be a MAX-2SAT formula with $n$ variables and $m$ clauses. 

\textit{Step 1 (Sawtooth folding).} Define the tent function $g(x) = 2\min(x, 1-x)$, which equals $2x$ for $x \in [0, 1/2]$ and $2-2x$ for $x \in [1/2, 1]$. This function can be implemented as $g(x) = \mathrm{ReLU}(2x) - 2\,\mathrm{ReLU}(2 x - 1)$, requiring depth 1 and width 2. The $n$-fold composition $T_n = g \circ g \circ \cdots \circ g$ ($n$ times) creates a sawtooth function with $2^n$ linear pieces on $[0,1]$, where each piece has slope $\pm 2^n$ (see Fig.~\ref{fig:sawtooth}). By composing $n$ tent functions, $T_n$ can be expressed as a ReLU network with depth $O(n)$ and width $O(1)$.

\textit{Step 2 (Bit extraction).} From the sawtooth function, we construct $n$ binary signals $b_1(x), \ldots, b_n(x)$ that encode the index $j$ of the interval $x \in [j \cdot 2^{-n}, (j+1) \cdot 2^{-n})$ in $n$-bit \emph{reflected Gray code} (see Fig.~\ref{fig:bit_extraction}). Each signal is realized by the composition
\[
    b_k(x) = T_k\!\left( \mathrm{ReLU}(2\,T_{n-k}(x) - 1) \right), \quad k = 1, \ldots, n,
\]
where the inner sawtooth $y:=T_{n-k}(x)$ partitions $[0,1]$ into $2^{n-k}$ pieces, the ramp $\mathrm{ReLU}(2y-1)$ extracts a triangular spike from each peak of $T_{n-k}$ (producing $2^{n-k-1}$ spikes whose input traverses $0{\to}1{\to}0$), and the outer $T_k$ folds every spike into $2^{k}$ sub-peaks. The resulting $b_k$ has $2^{n-1}$ unit-height peaks aligned with the intervals whose Gray-code index has bit $k$ set, and inherits slope $\pm 2^{n+1}$ on every linear piece (the extra factor $2$ relative to $T_n$ comes from the ramp $\mathrm{ReLU}(2y-1)$). Each $b_k$ has depth $O(n)$ and width $O(1)$, giving a total depth of $O(n)$ and width $O(n)$ for extracting all $n$ bits. The Gray-code ordering is an artifact of the composition order and is immaterial since Step~3 enumerates assignments by brute force.

\emph{Spike alignment.} Index the atomic sub-intervals by $I_j = [j\cdot 2^{-n},(j+1)\cdot 2^{-n})$, $j=0,\ldots,2^n-1$, and let $G(j)\in\{0,1\}^n$ denote the Gray-code representation of $j$. Define the common ``atomic tent'' $\tilde T_j(x) = 2^{n+1}\min(x - j\cdot 2^{-n},\; (j+1)\cdot 2^{-n} - x)$ on $I_j$, which is the unique tent of width $2^{-n}$, height $1$, and slope $\pm 2^{n+1}$ supported on $I_j$. From the composition $b_k = T_k\circ\mathrm{ReLU}\circ(2T_{n-k}-1)$, induction on $n-k$ yields
\begin{equation}\label{eq:bk_alignment}
    b_k(x)\big|_{I_j} = \begin{cases} \tilde T_j(x), & \text{if } G(j)_k = 1, \\ 0, & \text{if } G(j)_k = 0, \end{cases}
\end{equation}
as illustrated for $n=3$ in Fig.~\ref{fig:bit_extraction}. The same identity, with $G(j)_k$ flipped, holds for $\neg b_k$. Consequently, every Boolean combination $C(x)$ of literals formed by AND/OR/NOT satisfies $C(x)|_{I_j}\in\{\tilde T_j(x),0\}$, since on each $I_j$ all inputs reduce to either $\tilde T_j$ or $0$ and $\max/\min$ preserves this dichotomy.

\textit{Step 3 (Boolean operations).} Since each bit signal $b_k(x)$ shares the same piecewise-linear scaffold (slope $\pm 2^{n+1}$, $2^{n-1}$ peaks), Boolean operations preserve this structure:
\begin{itemize}
    \item NOT: a negated literal $\neg x_k$ is realized by flipping the inner ramp to $\mathrm{ReLU}(1 - 2\,T_{n-k}(x))$, which yields peaks on the complementary set of intervals while keeping the same slope magnitude $2^{n+1}$.
    \item OR: $b_i \lor b_j = \max(b_i, b_j) = (b_i + b_j + |b_i - b_j|)/2$, where $|z| = \max(z, 0) + \max(-z, 0)$.
    \item AND: $b_i \land b_j = \min(b_i, b_j) = (b_i + b_j - |b_i - b_j|)/2$.
\end{itemize}
Each gate requires depth $O(1)$ and width $O(1)$. For each clause $C_j = (l_{j1} \lor l_{j2})$, we compute its truth value using the OR operation. The total number of satisfied clauses is then computed by summing all clause outputs: $F_\phi(x) = \sum_{j=1}^{m} C_j(x)$, which requires depth $O(1)$ and width $O(m)$.

Finally, the constructed network $F_\phi$ requires depth $O(n)$ and width $O(n + m)$, with all weights and biases in $\{0, \pm 1, \pm 2\}$ (hence $O(1)$ bits each). Fix any atomic interval $I_j$ corresponding to the assignment $s = G(j)\in\{0,1\}^n$. By the alignment property \eqref{eq:bk_alignment} extended to literals and Boolean combinations, each clause output satisfies $C_l(x)|_{I_j} = q_l(s)\cdot \tilde T_j(x)$, where $q_l(s)\in\{0,1\}$ indicates whether $s$ satisfies clause $C_l$. Summing over $l=1,\ldots,m$ gives $F_\phi(x)|_{I_j} = q(s)\cdot \tilde T_j(x)$, where $q(s) := \sum_{l=1}^m q_l(s)$ is the number of clauses satisfied by~$s$.
Since $\tilde T_j$ has slope $\pm 2^{n+1}$ on $I_j$, the Lipschitz constant of $F_\phi$ on $I_j$ equals $q(s)\cdot 2^{n+1}$, and taking the maximum over all atomic intervals (equivalently, over all assignments $s$) yields $\mathrm{Lip}(F_\phi) = 2^{n+1}\cdot \mathrm{MAX\text{-}2SAT}(\phi)$. Thus, computing $\mathrm{Lip}(F_\phi)$ determines the optimal value of MAX-2SAT. Although $\mathrm{Lip}(F_\phi)$ is exponential in $n$, the encoding size of the network (number of neurons and bit-length of parameters) is polynomial in $n$ and $m$, ensuring a valid polynomial-time reduction. Since MAX-2SAT is APX-hard with inapproximability threshold $22/21$ \cite{haastad2001some}, computing the Lipschitz constant of a 1D ReLU network is also NP-hard and APX-hard within the same threshold. A reference implementation of the construction is available at \url{https://github.com/lixc21/MAX2-SAT-to-NN}.
\end{proof}

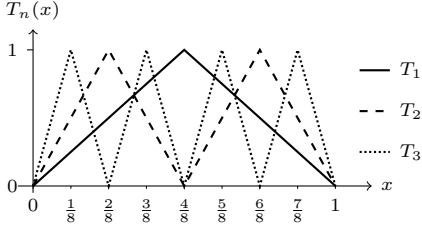
\begin{figure}[!t]
\centering
\begin{tikzpicture}[x=4cm, y=1.8cm]
  \draw[->] (-0.05,0) -- (1.12,0) node[right] {\scriptsize $x$};
  \draw[->] (0,-0.05) -- (0,1.15) node[above] {\scriptsize $T_n(x)$};
  
  \draw[thick, black] (0,0) -- (0.5,1) -- (1,0);
  
  \draw[thick, black, dashed] (0,0) -- (0.25,1) -- (0.5,0) -- (0.75,1) -- (1,0);
  
  \draw[thick, black, densely dotted] 
    (0,0) -- (0.125,1) -- (0.25,0) -- (0.375,1) -- (0.5,0) -- 
    (0.625,1) -- (0.75,0) -- (0.875,1) -- (1,0);
  
  \node[below] at (0,-0.03) {\scriptsize $0$};
  \node[below] at (0.125,-0.03) {\scriptsize $\frac{1}{8}$};
  \node[below] at (0.25,-0.03) {\scriptsize $\frac{2}{8}$};
  \node[below] at (0.375,-0.03) {\scriptsize $\frac{3}{8}$};
  \node[below] at (0.5,-0.03) {\scriptsize $\frac{4}{8}$};
  \node[below] at (0.625,-0.03) {\scriptsize $\frac{5}{8}$};
  \node[below] at (0.75,-0.03) {\scriptsize $\frac{6}{8}$};
  \node[below] at (0.875,-0.03) {\scriptsize $\frac{7}{8}$};
  \node[below] at (1,-0.03) {\scriptsize $1$};
  \foreach \x in {0, 0.125, 0.25, 0.375, 0.5, 0.625, 0.75, 0.875, 1} {
    \draw (\x,0) -- (\x,-0.02);
  }
  
  \node[left] at (-0.02,0) {\scriptsize $0$};
  \node[left] at (-0.02,1) {\scriptsize $1$};
  \draw (0,1) -- (-0.02,1);
  
  \draw[thick, black] (1.08,0.85) -- (1.18,0.85) node[right, font=\scriptsize] {$T_1$};
  \draw[thick, black, dashed] (1.08,0.55) -- (1.18,0.55) node[right, font=\scriptsize] {$T_2$};
  \draw[thick, black, densely dotted] (1.08,0.25) -- (1.18,0.25) node[right, font=\scriptsize] {$T_3$};
\end{tikzpicture}
\caption{Sawtooth functions $T_n(x)$ for $n=1,2,3$. The tent function $T_1$ has 2 linear pieces. The composition $T_2 = g \circ T_1$ has 4 pieces. The function $T_3$ has 8 pieces. In general, $T_n$ has $2^n$ pieces, partitioning $[0,1]$ into $2^n$ intervals that encode all $n$-bit binary strings.}
\label{fig:sawtooth}
\vspace{-0.4cm}
\end{figure}

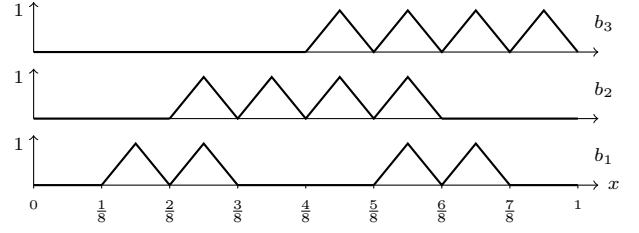
\begin{figure}[!t]
\centering
\begin{tikzpicture}[x=0.9cm,y=0.55cm]
  \node[anchor=west] at (8.1,2.3) {\scriptsize $b_3$};
  \draw[->] (0,1.6) -- (8.3,1.6);
  \draw[->] (0,1.6) -- (0,2.8);
  \node[left] at (0,2.6) {\scriptsize $1$};
  \draw[thick, black] (0,1.6) -- (4,1.6); 
  \draw[thick, black] (4,1.6) -- (4.5,2.6) -- (5,1.6) -- (5.5,2.6) -- (6,1.6) -- (6.5,2.6) -- (7,1.6) -- (7.5,2.6) -- (8,1.6); 
  
  \node[anchor=west] at (8.1,0.7) {\scriptsize $b_2$};
  \draw[->] (0,0) -- (8.3,0);
  \draw[->] (0,0) -- (0,1.2);
  \node[left] at (0,1.0) {\scriptsize $1$};
  \draw[thick, black] (0,0) -- (2,0); 
  \draw[thick, black] (2,0) -- (2.5,1.0) -- (3,0) -- (3.5,1.0) -- (4,0) -- (4.5,1.0) -- (5,0) -- (5.5,1.0) -- (6,0); 
  \draw[thick, black] (6,0) -- (8,0); 
  
  \node[anchor=west] at (8.1,-0.9) {\scriptsize $b_1$};
  \draw[->] (0,-1.6) -- (8.3,-1.6) node[right] {\scriptsize $x$};
  \draw[->] (0,-1.6) -- (0,-0.4);
  \node[left] at (0,-0.6) {\scriptsize $1$};
  \draw[thick, black] (0,-1.6) -- (1,-1.6); 
  \draw[thick, black] (1,-1.6) -- (1.5,-0.6) -- (2,-1.6) -- (2.5,-0.6) -- (3,-1.6); 
  \draw[thick, black] (3,-1.6) -- (5,-1.6); 
  \draw[thick, black] (5,-1.6) -- (5.5,-0.6) -- (6,-1.6) -- (6.5,-0.6) -- (7,-1.6); 
  \draw[thick, black] (7,-1.6) -- (8,-1.6); 
  
  \node[below] at (0,-1.75) {\tiny $0$};
  \node[below] at (1,-1.75) {\tiny $\frac{1}{8}$};
  \node[below] at (2,-1.75) {\tiny $\frac{2}{8}$};
  \node[below] at (3,-1.75) {\tiny $\frac{3}{8}$};
  \node[below] at (4,-1.75) {\tiny $\frac{4}{8}$};
  \node[below] at (5,-1.75) {\tiny $\frac{5}{8}$};
  \node[below] at (6,-1.75) {\tiny $\frac{6}{8}$};
  \node[below] at (7,-1.75) {\tiny $\frac{7}{8}$};
  \node[below] at (8,-1.75) {\tiny $1$};
  \foreach \x in {0,1,2,3,4,5,6,7,8} {
    \draw (\x,-1.6) -- (\x,-1.68);
  }
\end{tikzpicture}
\caption{Bit extraction from the sawtooth function $T_3(x)$ for $n=3$ bits. The input domain $[0,1]$ is divided into $8$ intervals indexed by $j=0,\ldots,7$. Each row shows a bit signal $b_k = T_k(\mathrm{ReLU}(2 T_{n-k}(x) - 1))$, encoding the $k$-th bit of the reflected Gray code of the interval index $j$: $b_1$ (LSB) peaks on $j\in\{1,2,5,6\}$, $b_2$ on $j\in\{2,3,4,5\}$, and $b_3$ (MSB) on $j\in\{4,5,6,7\}$. Flat segments correspond to bit value 0; peaks correspond to bit value 1.}
\label{fig:bit_extraction}
\vspace{-0.4cm}
\end{figure}

Next, to transfer hardness results from ReLU networks to mpQPs, we show that any ReLU network can be represented as an mpQP with polynomial overhead. To the best of our knowledge, this construction is new and enables a direct transfer of the complexity results.

\begin{lemma}[mpQP Representation of ReLU Networks]\label{lem:mpqp_relu_equiv}
Let $F:\mathbb{R}^{n_x}\to\mathbb{R}^{d_{D+1}}$ be a $D$-layer ReLU network with $N$ hidden neurons and output $F(x)=W_{D+1}h_D(x)+b_{D+1}$, defined on $\|x\|_\infty\le B_0$. Then there exists a polynomial-time construction of an mpQP with $N+d_{D+1}$ decision variables, $2N$ inequality constraints and $d_{D+1}$ equality constraints, whose unique optimal solution decomposes as $z^*(x)=(z^*_1(x),\ldots,z^*_D(x),\,z^*_{D+1}(x))$ with $z^*_\ell\in\mathbb{R}^{d_\ell}$ and $z^*_{D+1}(x)=F(x)$.
\end{lemma}

\begin{proof}
We construct an mpQP whose optimal solution reproduces the ReLU network output, and verify this by checking the KKT conditions. Define the mpQP with parameter $z_0 = x$ as:
\begin{equation}\label{eq:relu_mpqp_construction}
\begin{aligned}
    \min_{z_1, \ldots, z_D,\, z_{D+1}} \quad & \sum_{\ell=1}^{D} \frac{1}{2}\|z_\ell + M_\ell \mathbf{1}\|^2 + \frac{1}{2}\|z_{D+1}\|^2 \\
    \text{s.t.} \quad & z_\ell \ge 0, \quad \ell = 1, \ldots, D, \\
    & z_\ell \ge W_\ell z_{\ell-1} + b_\ell, \quad \ell = 1, \ldots, D, \\
    & z_{D+1} = W_{D+1} z_D + b_{D+1}.
\end{aligned}
\end{equation}
The last equality realises the network's affine output layer. The objective is strictly convex, so any KKT point is the unique global minimum. Introducing non-negative multipliers $\alpha_\ell, \beta_\ell \ge 0$ for the hidden-layer constraints $z_\ell \ge 0$ and $z_\ell \ge W_\ell z_{\ell-1} + b_\ell$ ($\ell=1,\ldots,D$), and an unconstrained multiplier $\mu\in\mathbb{R}^{d_{D+1}}$ for the output-layer equality, the Lagrangian is
\[
\begin{aligned}
\mathcal{L} = {}& \sum_{\ell=1}^{D} \Bigl[ \tfrac{1}{2}\|z_\ell + M_\ell \mathbf{1}\|^2 - \alpha_\ell^\top z_\ell - \beta_\ell^\top (z_\ell - W_\ell z_{\ell-1} - b_\ell) \Bigr] \\
& + \tfrac{1}{2}\|z_{D+1}\|^2 - \mu^\top(z_{D+1} - W_{D+1} z_D - b_{D+1}).
\end{aligned}
\]
The KKT conditions require primal feasibility, dual feasibility ($\alpha_\ell,\beta_\ell\ge 0$; $\mu$ free), complementary slackness ($\alpha_{\ell,j} z_{\ell,j} = 0$ and $\beta_{\ell,j} (z_{\ell,j} - (W_\ell z_{\ell-1} + b_\ell)_j) = 0$), and stationarity. Setting $\beta_{D+1}:=\mu$, differentiation with respect to $z_\ell$ for $\ell=1,\ldots,D$ gives the uniform identity
\begin{equation}\label{eq:stationarity_kkt}
z_\ell + M_\ell \mathbf{1} - \alpha_\ell - \beta_\ell + W_{\ell+1}^\top \beta_{\ell+1} = 0,
\end{equation}
while differentiation with respect to $z_{D+1}$ gives $z_{D+1}-\mu=0$, i.e., $\beta_{D+1}=\mu=z_{D+1}^*$ (which is signed, in contrast to $\beta_\ell\ge 0$ for $\ell\le D$).

To ensure valid non-negative multipliers exist, we determine the constants $M_\ell$ via a two-pass procedure. First, compute activation bounds $U_\ell$ by forward propagation: set $U_0 = B_0 \mathbf{1}$ where $\|x\|_\infty \le B_0$, and recursively define $U_\ell = \mathrm{ReLU}(|W_\ell| U_{\ell-1} + |b_\ell|)$. These bounds satisfy $0 \le h_\ell \le U_\ell$ component-wise for all valid inputs. Second, compute $M_\ell$ and auxiliary dual bounds $\bar{\beta}_\ell$ by backward recursion from layer $D+1$ to layer $1$:
\begin{align}
    \bar{\beta}_{D+1} &= |W_{D+1}|\,U_D + |b_{D+1}|, \nonumber \\
    M_\ell &= \max\!\left(1,\, \bigl\| |W_{\ell+1}^\top|\, \bar{\beta}_{\ell+1} \bigr\|_\infty + 1 \right), \label{eq:M_def} \\
    \bar{\beta}_\ell &= U_\ell + M_\ell \mathbf{1} + |W_{\ell+1}^\top|\, \bar{\beta}_{\ell+1}, \nonumber
\end{align}
where $|A|$ denotes the element-wise absolute value of a matrix. The initialisation $\bar\beta_{D+1}$ is a componentwise upper bound on $|F(x)|=|z_{D+1}^*|$ obtained by one forward pass.

We now verify that the ReLU network output $z_\ell^* = h_\ell(x)$ satisfies the KKT conditions. Primal feasibility holds by definition of the ReLU function, and the output-layer equality is automatically satisfied by $z_{D+1}^*=W_{D+1}h_D(x)+b_{D+1}=F(x)$. For the dual variables, \eqref{eq:stationarity_kkt} gives $\alpha_\ell + \beta_\ell = z_\ell + M_\ell \mathbf{1} + W_{\ell+1}^\top \beta_{\ell+1}$. We show by backward induction on $\ell=D,D-1,\ldots,1$ that valid multipliers exist with $|\beta_\ell|\le \bar\beta_\ell$ component-wise (and $\beta_\ell\ge 0$ for $\ell\le D$).

For the base case $\ell = D+1$, the output stationarity forces $\beta_{D+1}=z_{D+1}^*=F(x)$, and $|\beta_{D+1}|\le |W_{D+1}|U_D+|b_{D+1}|=\bar\beta_{D+1}$ by the forward bound.

For $\ell\le D$, assume $|\beta_{\ell+1}|\le \bar{\beta}_{\ell+1}$ component-wise. Consider component $j$. If $(W_\ell h_{\ell-1} + b_\ell)_j > 0$, then $z_{\ell,j} = (W_\ell h_{\ell-1} + b_\ell)_j > 0$, and complementarity requires $\alpha_{\ell,j} = 0$. Stationarity then gives $\beta_{\ell,j} = z_{\ell,j} + M_\ell + (W_{\ell+1}^\top \beta_{\ell+1})_j$. Using $|(W_{\ell+1}^\top\beta_{\ell+1})_j|\le (|W_{\ell+1}^\top|\bar\beta_{\ell+1})_j\le \||W_{\ell+1}^\top|\bar\beta_{\ell+1}\|_\infty$ together with $M_\ell \ge \||W_{\ell+1}^\top|\bar\beta_{\ell+1}\|_\infty + 1$ from \eqref{eq:M_def}, we obtain $\beta_{\ell,j}\ge z_{\ell,j}+1 > 0$ and $\beta_{\ell,j} \le U_{\ell,j} + M_\ell + (|W_{\ell+1}^\top|\bar\beta_{\ell+1})_j = \bar\beta_{\ell,j}$. Conversely, if $(W_\ell h_{\ell-1} + b_\ell)_j \le 0$, then $z_{\ell,j} = 0$, complementarity requires $\beta_{\ell,j} = 0$ (the second constraint is strictly slack when the inequality is strict, and we choose $\beta_{\ell,j}=0$ at equality), and stationarity gives $\alpha_{\ell,j} = M_\ell + (W_{\ell+1}^\top \beta_{\ell+1})_j \ge 1 > 0$ by the same bound on $M_\ell$.

Thus the specific choice of $M_\ell$ guarantees stationarity with non-negative multipliers respecting complementarity at every hidden layer. Since the mpQP is strictly convex, these KKT conditions are necessary and sufficient for global optimality, proving $z_\ell^* = h_\ell(x)$ for $\ell=1,\ldots,D$ and $z_{D+1}^* = W_{D+1} h_D(x) + b_{D+1} = F(x)$, so the network output is exactly recovered from the mpQP solution. Finally, the construction parameters $\{M_\ell\}$ are derived via a fixed number of forward and backward passes; their values may grow exponentially with depth, but their binary representations have length polynomial in the size of the network parameters. A reference implementation of the construction is available at \url{https://github.com/lixc21/NN-to-mpQP}.
\end{proof}

The two reductions can now be composed. Lemma~\ref{lem:1d_nn_hard} reduces MAX-2SAT to 1D-NN-Lipschitz, and Lemma~\ref{lem:mpqp_relu_equiv} embeds any ReLU network into an mpQP whose output-layer subvector reproduces $F(x)$. Their composition transfers the hardness to the scalar-parameter mpQP. One subtlety remains: the mpQP Lipschitz constant depends on the full solution vector $z^*$, not only on its output component. The next theorem closes this gap via an output-layer scaling argument, and yields the desired NP- and APX-hardness.

\begin{theorem}\label{thm:scalar_hard}
Computing the Lipschitz constant $L$ of an mpQP is NP-hard and APX-hard even when $n_x = 1$. 
\end{theorem}

\begin{proof}
Apply Lemma~\ref{lem:mpqp_relu_equiv} to the ReLU network $F$ from Lemma~\ref{lem:1d_nn_hard}. The resulting mpQP has solution vector $z^*(x)=[z_1^*(x)^\top,\ldots,z_{D+1}^*(x)^\top]^\top$ with $z_{D+1}^*(x)=F(x)$, and its Lipschitz constant is the Euclidean norm of the full Jacobian, not of the output component alone. Let $P_{\mathrm{rest}}$ denote the Lipschitz constant of the hidden components $z_{1:D}^*$; the bit-bounded weights of Lemma~\ref{lem:1d_nn_hard} give $P_{\mathrm{rest}}=O(2^{\mathrm{poly}(n)})$.

\emph{Step 1: Output-layer scaling.} Fix a parameter $\delta\in(0,1)$ to be chosen below. We scale the output layer by a factor $K$, defining $\tilde{F}(x) = K F(x)$ (equivalently, replacing $W_{D+1}, b_{D+1}$ by $K W_{D+1}, K b_{D+1}$ in Lemma~\ref{lem:mpqp_relu_equiv}). This transforms the output component to $\tilde{z}_{D+1}^*(x) = K z_{D+1}^*(x)$, while leaving $z_{1:D}^*$ unaffected. The Lipschitz constant of the modified solution vector $\tilde{z}^*$ then satisfies
\begin{equation}\label{eq:scalar_sandwich_raw}
    K \cdot \mathrm{Lip}(F) \;\le\; \mathrm{Lip}(\tilde{z}^*) \;\le\; \sqrt{P_{\mathrm{rest}}^2 + K^2 \cdot \mathrm{Lip}(F)^2}.
\end{equation}

\emph{Step 2: Multiplicative sandwich.} Let $n$ and $m$ denote the variable and clause counts of the underlying MAX-2SAT instance. Any non-trivial instance has $\mathrm{MAX\text{-}2SAT}(\phi)\ge 1$, so Lemma~\ref{lem:1d_nn_hard} gives $\mathrm{Lip}(F)=2^{n+1}\cdot\mathrm{MAX\text{-}2SAT}(\phi)\ge 2^{n+1}$. Choose $K=\bigl\lceil P_{\mathrm{rest}}/(2^{n+1}\delta)\bigr\rceil$, so that $P_{\mathrm{rest}}/(K\cdot\mathrm{Lip}(F))\le \delta$. Dividing \eqref{eq:scalar_sandwich_raw} by $K\cdot\mathrm{Lip}(F)$ yields the sandwich
\begin{equation}\label{eq:scalar_sandwich}
    1 \;\le\; \frac{\mathrm{Lip}(\tilde{z}^*)}{K\cdot\mathrm{Lip}(F)} \;\le\; \sqrt{1+\delta^2} \;\le\; 1+\delta.
\end{equation}
Since $\log K = \log P_{\mathrm{rest}} + (n+1) + \log(1/\delta) = \mathrm{poly}(n)$ for any $\delta\ge 2^{-\mathrm{poly}(n)}$, the binary size of $K$ remains polynomial, and so does the entire mpQP description (the $M_\ell$ in Lemma~\ref{lem:mpqp_relu_equiv} scale at most by a factor $\mathrm{poly}(K)$, contributing only $\mathrm{poly}(n)$ extra bits).

\emph{Step 3: Gap-preserving reduction.} Suppose, for contradiction, that for some fixed $\rho<22/21$ a polynomial-time algorithm outputs a rational $\widehat{L}$ satisfying $\widehat{L}\le \mathrm{Lip}(\tilde{z}^*)\le \rho\widehat{L}$ in the sense of Problem~\ref{prob:lip_approx} (the threshold $22/21$ is H\aa{}stad's MAX-2SAT inapproximability factor \cite{haastad2001some} and differs from $\sqrt{17/16}$ in Theorem~\ref{thm:complexity_main} because the latter reduction operates on $L^2$). Set $\widehat{L}_F := \widehat{L}/\bigl(K(1+\delta)\bigr)$. Combining $\widehat L\le \mathrm{Lip}(\tilde z^*)$ with the upper inequality of \eqref{eq:scalar_sandwich} yields $\widehat L_F\le \mathrm{Lip}(F)$, while combining $\mathrm{Lip}(\tilde z^*)\le \rho\widehat L$ with the lower inequality of \eqref{eq:scalar_sandwich} yields $\mathrm{Lip}(F)\le \rho\widehat L/K = \rho(1+\delta)\,\widehat L_F$. Hence
\begin{equation}\label{eq:scalar_gap_clean}
    \widehat{L}_F \;\le\; \mathrm{Lip}(F) \;\le\; \rho(1+\delta)\,\widehat{L}_F,
\end{equation}
so $\widehat L_F$ is a $\rho(1+\delta)$-approximation of $\mathrm{Lip}(F)$ in the single-sided sense of Problem~\ref{prob:lip_approx}. Choosing $\delta=\tfrac12\bigl((22/21)/\rho-1\bigr)>0$ (a positive rational constant independent of $n,m$, hence of polynomial bit-length) gives $\rho(1+\delta)=(\rho+22/21)/2<22/21$, which by Lemma~\ref{lem:1d_nn_hard} is impossible unless P$=$NP. This establishes APX-hardness of Problem~\ref{prob:lip_approx} and NP-hardness of exactly computing $L$ with $n_x=1$.
\end{proof}

\section{Numerical Experiments}\label{sec:numerical}

We present numerical experiments that validate the theoretical analysis. All computations use \textsc{pdaqp}~\cite{arnstrom2022unifying}, a parametric dual active-set solver that enumerates all critical regions, on a 128-core server (AMD EPYC 7742, 251~GB RAM) with a per-instance timeout of 300~s.

For Experiments~1 and~2, we generate random mpQP instances of the form~\eqref{eq:mpqp_def} as follows: $Q = M^\top M + I_{n_z}$ with $M \sim \mathcal{N}(0,1)^{n_z\times n_z}$ (ensuring $Q \succ 0$), $H \sim \mathcal{N}(0,1)^{n_x \times n_z}$, $f = 0$, $A \sim \mathcal{N}(0,1)^{n_c \times n_z}$, $S \sim \mathcal{N}(0,1)^{n_c \times n_x}$, and $b = 2|\xi|+2$ with $\xi \sim \mathcal{N}(0,1)^{n_c}$ (ensuring feasibility near the origin). The parameter domain is $\mathcal{X}=[-5,5]^{n_x}$. For each configuration of $(n_x, n_z, n_c)$, we use five independent random seeds and report the mean number of CRs. Experiment~3 instead uses the deterministic tent-map ReLU network from Lemma~\ref{lem:1d_nn_hard}; its setup is detailed in Section~\ref{subsec:exp3}.

\begin{figure}[!t]
\centering
\includegraphics[width=0.8\linewidth]{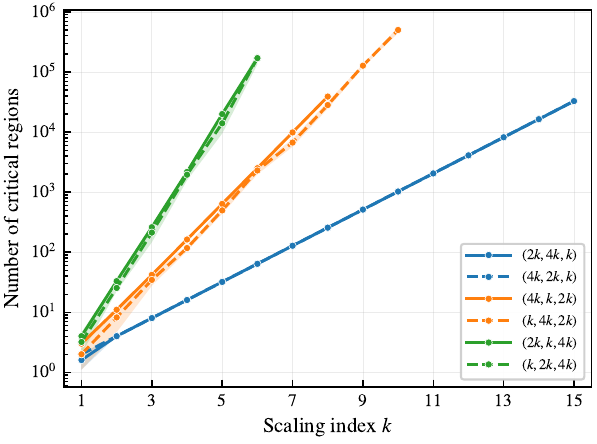}
\caption{Exp.~1: joint scaling $(n_x,n_z,n_c)=(a_x k,a_z k,a_c k)$. Number of CRs grows exponentially in $k$.}
\label{fig:exp1}
\vspace{-0.3cm}
\end{figure}

\begin{figure}[!t]
\centering
\includegraphics[width=0.8\linewidth]{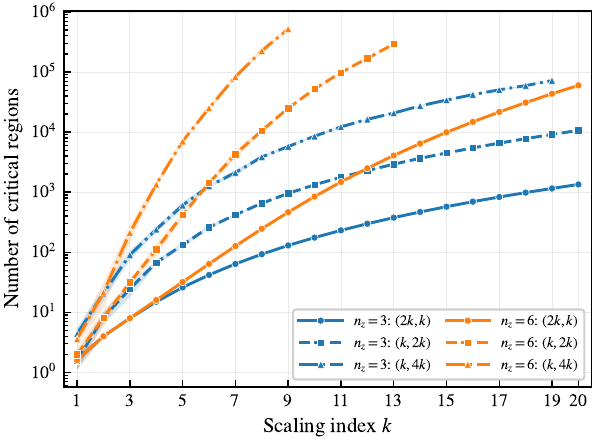}
\caption{Exp.~2A: fixed $n_z$. Polynomial growth as $O(n_c^{n_z})$.}
\label{fig:exp2a}
\vspace{-0.5cm}
\end{figure}

\begin{figure}[!t]
\centering
\includegraphics[width=0.8\linewidth]{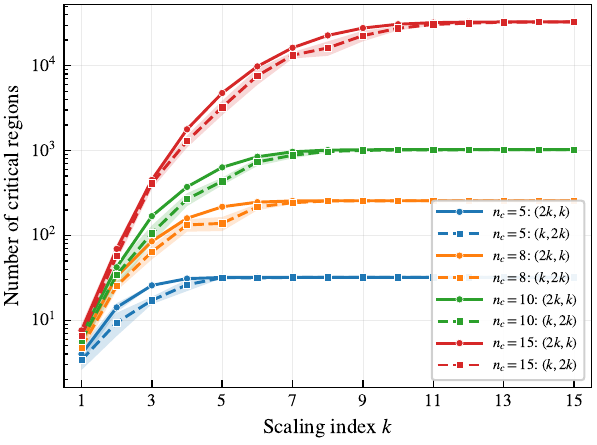}
\caption{Exp.~2B: fixed $n_c$. Curves saturate at $\le 2^{n_c}$.}
\label{fig:exp2b}
\vspace{-0.3cm}
\end{figure}

\begin{figure}[!t]
\centering
\includegraphics[width=0.8\linewidth]{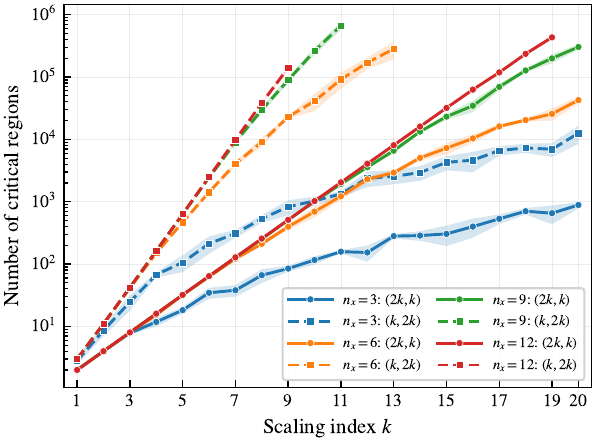}
\caption{Exp.~2C: fixed $n_x$. Exponential growth persists even with $n_x$ fixed.}
\label{fig:exp2c}
\vspace{-0.5cm}
\end{figure}

\subsection{Experiment 1: Joint Scaling}\label{subsec:exp1}

We let all three dimensions grow simultaneously as $(n_x,n_z,n_c)=(a_x k,\,a_z k,\,a_c k)$ and sweep $k=1,2,\ldots$, for six triples $(a_x,a_z,a_c)$ with pairwise distinct elements from $\{1,2,4\}$. Figure~\ref{fig:exp1} shows that all configurations exhibit exponential growth of the number of CRs, with the growth rate governed primarily by $n_z$ and $n_c$. Configurations with $a_c=4$ (green) exceed the timeout around $k=6$, while those with $a_c=1$ (blue) remain tractable up to $k=15$. The narrow shaded bands (mean $\pm$ one standard deviation) indicate that this exponential growth is consistent across random seeds.

\subsection{Experiment 2: Fixing One Dimension}\label{subsec:exp2}

To isolate the role of each dimension, we conduct three sub-experiments that fix one of $n_z$, $n_c$, or $n_x$ and let the remaining two grow with $k$.

\paragraph{Experiment 2A: Fixed $n_z$}
We fix $n_z \in \{3, 6\}$ and let $(n_x, n_c)$ grow with $k$. Figure~\ref{fig:exp2a} shows that the number of CRs still grows rapidly --- exceeding $10^5$ for the largest configurations --- even though $n_z$ is fixed. This is consistent with Lemma~\ref{lem:fixed_nc_nz_new}: while fixed $n_z$ yields polynomial-time solvability, the number of CRs can still grow as $O(n_c^{n_z})$.

\paragraph{Experiment 2B: Fixed $n_c$}
We fix $n_c \in \{5, 8, 10, 15\}$ and let $(n_x, n_z)$ grow with $k$. Figure~\ref{fig:exp2b} shows that the number of CRs \emph{saturates} to a plateau as $k$ increases: e.g., for $n_c=5$ the count stabilizes near 30, while for $n_c=15$ it levels off near $3\times10^4$. This is a direct consequence of Lemma~\ref{lem:fixed_nc_nz_new}: with $n_c$ fixed, there are at most $2^{n_c}$ possible active sets, so the number of CRs is bounded regardless of $n_z$ and $n_x$.

\paragraph{Experiment 2C: Fixed $n_x$}
We fix $n_x \in \{3, 6, 9, 12\}$ and let $(n_z, n_c)$ grow with $k$. Figure~\ref{fig:exp2c} shows that fixing $n_x$ \emph{does not} prevent exponential growth of the number of CRs --- all curves continue to climb steeply, consistent with Theorem~\ref{thm:scalar_hard}: the parameter dimension $n_x$ is not the source of complexity.

\subsection{Experiment 3: Tent-Map ReLU$\,\to\,$mpQP}\label{subsec:exp3}

We construct a $D$-layer tent-map ReLU network ($n_x=1$, width~2) as in Lemma~\ref{lem:1d_nn_hard}, for $D\in\{1,\ldots,8\}$, and convert it to an mpQP via Lemma~\ref{lem:mpqp_relu_equiv} on the parameter domain $\mathcal{X}=[-1,1]$. Since the scalar linear output layer is absorbed into the objective, the output variable $z_{D+1}$ and its equality constraint are eliminated, leaving only the hidden activations as decision variables; with $N=2D$ this gives $n_z=2D$ (one fewer than the worst-case $N+d_{D+1}$ in Lemma~\ref{lem:mpqp_relu_equiv}) and $n_c=4D$ (matching the $2N$ inequalities in Lemma~\ref{lem:mpqp_relu_equiv}). Table~\ref{tab:exp3} reports the number of CRs enumerated by the solver alongside the theoretical count $2^D+1$ (the tent-map network partitions $[0,1]$ into $2^D$ linear pieces, plus one exterior region on $[-1,0]$ where the network is identically zero). For $D\le 6$, the two values match exactly. For $D=7$ and $D=8$, the actual counts (134 and 266) slightly exceed $2^D+1$ (129 and 257), because the mpQP construction may introduce additional CRs that do not correspond to distinct ReLU activation patterns. The $\Theta(2^D)$ scaling is clearly maintained even at $n_x=1$, consistent with Theorem~\ref{thm:scalar_hard}: the parameter dimension is not the source of complexity.

\begin{table}[!t]
\centering
\caption{Tent-map ReLU network ($n_x=1$, width~2) converted to mpQP.}\label{tab:exp3}
\begin{tabular}{ccccc}
\toprule
Depth $D$ & $n_z$ & $n_c$ & Theoretical CRs & Actual CRs \\
\midrule
1 & 2 & 4 & 3 & 3 \\
2 & 4 & 8 & 5 & 5 \\
3 & 6 & 12 & 9 & 9 \\
4 & 8 & 16 & 17 & 17 \\
5 & 10 & 20 & 33 & 33 \\
6 & 12 & 24 & 65 & 65 \\
7 & 14 & 28 & 129 & 134 \\
8 & 16 & 32 & 257 & 266 \\
\bottomrule
\end{tabular}
\vspace{-0.5cm}
\end{table}

\section{Conclusion}
This paper gives a complete complexity characterization for computing the Lipschitz constant of mpQP solution maps: the problem is NP-hard and APX-hard in general and remains so even when the parameter is scalar, while fixing either the number of constraints or the number of decision variables restores polynomial-time solvability. The scalar-parameter hardness relies on a new constructive equivalence recovering any ReLU network as the unique optimal solution of an mpQP, which is of independent interest for transferring complexity and verification results between the two models. Numerical experiments confirm that the joint growth of the constraint and decision-variable dimensions is the source of intractability.

\bibliographystyle{IEEEtran}
\bibliography{mybib.bib}

\end{document}